\documentclass[11pt]{article}
\usepackage{latexsym,amsmath,amsfonts,amsthm}

\newtheorem{theorem}[]{Theorem}

\theoremstyle{definition}

\newtheorem{problem}[]{Problem}

\newtheorem{proposition}[]{Proposition}
\newtheorem{edgerule}[]{Rule}

\begin{document}

\begin{center}
{\LARGE\textbf{A Note on the Hardness of Graph Diameter Augmentation Problems}}\\
\bigskip
\bigskip
James Nastos\footnote{email: \texttt{jnastos@interchange.ubc.ca}}, Yong Gao\footnote{email: \texttt{yong.gao@ubc.ca}}\\
\smallskip
\small\emph{Department of Computer Science, Irving K. Barber School of Arts and Sciences.}
\smallskip
\small\emph{University of British Columbia Okanagan, Kelowna, Canada V1V 1V7}
\end{center}

\begin{abstract}
A graph has \emph{diameter} $D$ if every pair of vertices are connected by a path of at most $D$ edges. The Diameter-$D$ Augmentation problem asks how to add the a number of edges to a graph in order to make the resulting graph have diameter $D$. It was previously known that this problem is NP-hard \cite{GJ}, even in the $D=2$ case. In this note, we give a simpler reduction to arrive at this fact and show that this problem is W[2]-hard.
\bigskip


\end{abstract}

\bigskip
\normalsize

\bigskip
\normalsize

\section{Introduction}

\noindent
A graph $G$ has \emph{diameter} $D$ if every pair of vertices are connected by a path of at most $D$. The {\sc Graph Diameter-D Augmentation} problem takes as input a graph $G=(V,E)$ and a value $k$ and asks whether there exists a set $E_2$ of new edges so that the graph $G_2 = (V, E \cup E_2)$ has diameter $D$. This problem was known to be NP-hard for $D \geq 3$ \cite{SBL} and was later shown to remain hard for the $D=2$ case \cite{LMS}. The proof in \cite{LMS} reduced a restricted (but still NP-hard \cite{GJ}) {\sc 3-Sat} problem to a relaxed dominating set problem (which they called {\sc Semi-Dominating Set}) which was then reduced to {\sc Diameter-2 Augmentation}. In this note, we provide a reduction to {\sc Diameter-2 Augmentation} directly from {\sc Dominating Set}, which not only provides a cleaner proof of NP-hardness but also establishes that {\sc Diameter-2 Augmentation} is W[2]-hard.

\bigskip

\noindent
An algorithm is called \emph{fixed-parameter tractable} (or FPT) if its runtime is $O(f(k)n^c)$ where $n$ is the input size, $f$ is a function of $k$ which does not depend on $n$ and $c$ is a constant. When the value $k$ is fixed, this is essentially a polynomial runtime, and in particular for any fixed $k$ it is the same polynomial (up to coefficients.) FPT algorithms have received much attention lately as many NP-hard problems have been shown to be fixed-parameter tractable. For instance, the {\sc Vertex Cover} problem has an algorithm (\cite{CKX}) running in $O(1.2738^k + kn)$ which is linear in $n$ for any fixed $k$. Analogous to the idea of NP-hardness, there is a measure of hardness for parameterized problems which depend on parameterized reductions. Some well-known parameterized-hard problems are {\sc Clique} (which is W[1]-hard) and {\sc Dominating Set} (which is W[2]-hard). These results and a thorough introduction to parameterized problems can be found in \cite{N}. Being parameterized-hard also has implications for the approximatibility of the problem: namely, a problem which is W[1]-hard is unlikely to have an efficient polynomial-time approximation scheme (EPTAS) \cite{M}.

\bigskip

\noindent

\subsection{The Reduction}

We proceed with a reduction from the parameterized dominating set problem to the parameterized diameter-2 augmentation problem after a formal description of each of these problems and of what constitutes a parameterized reduction. In this report, we consider input graphs which are connected.

\begin{problem}{\sc Dominating Set}
\\{\sc Input:} A graph $G=(V,E)$ and a positive integer $k$.
\\{\sc Task:} To determine if there exists a set $S \subseteq V$ of size at most $k$ such that for every $v \in V\setminus S$ there is some $s \in S$ where $\{s, v\}$ is an edge.
\end{problem}

\begin{problem}{\sc Diameter-2 Augmentation}
\\{\sc Input:} A graph $G=(V,E)$ and a positive integer $k$.
\\{\sc Task:} To determine if there exists a set of at most $k$ edges that can be added to $G$ so that the resulting graph has diameter 2.
\end{problem}

We must reduce {\sc Dominating Set} to {\sc Diameter-2 Augmentation} via a \emph{parameterized reduction}. That is, we must give a mapping that sends a yes-instance $(G_1,k_1)$ of {\sc Dominating Set} to a yes-instance $(G_2,k_2)$ of {\sc Diameter-2 Augmentation} where $k_2$ depends on $k_1$ alone. We will provide a mapping here where $k_2 = k_1$.

Let $(G_1,k_1)$ be an instance of {\sc Dominating Set}, where $G_1 = (V_1,E_1)$. We construct a graph $G_2$ with two copies of $G_1$ called $U_1$ and $U_2$. Any two vertices $u_1 \in U_1$ and $u_2 \in U_2$ that correspond to the same vertex $v \in V_1$ will be called \emph{twins}. For each vertex $w$ in $U_1$, join an edge between $w$ and its twin in $U_2$. Let $w_i$ and $w_j$ be any two distinct vertices in $U_1 \cup U_2$. In $G_2$, create a new set $Y$ of vertices $y(w_1,w_2)$ such that $Y$ induces a complete graph and each vertex $y(w_1,w_2)$ is adjacent to $w_1$ and to $w_2$. Finally, we create in $G_2$ a vertex $z$ adjacent to every vertex of $Y$ and adjacent to no vertex in $U_1 \cup U_2$, and create a vertex $x$ adjacent to $z$ alone.

Note that $G_2$ has diameter at most 3. Every pair of vertices in $G_2$ which is not connected by a 2-path must be $x$ with some $w_i \in U_1 \cup U2$. It is easy to see that if a dominating set $D$ of $G_1$ contained $k$ vertices, then the set of edges $\{x, d\}, d\in D$ forms a diameter-2 augmenting set (also of size $k$) for $G_2$. We now prove the converse.

\begin{theorem} $G_1$ has a dominating set of size $k$ if and only if $G_2 = (V_2,E_2)$ has an augmenting set of edges $S$ such that $H = (V_2, E_2 \cup S)$ has diameter 2.
\end{theorem}

\begin{proof} Given a $k$-augmenting set of $G_2$, we will construct a dominating set $D$ of $G_1$ also of size $k$. If an augmenting set of $G_2$ only contains edges from $x$ to vertices in $U_1$ we will call it \emph{proper.} We can extract a dominating set of $U_1$ (and thus of $G_1$) from a proper diameter-2 augmenting set $S$ of $G_2$ simply by taking all the vertices of $U_1$ that are adjacent to $x$ in $S$.

Say that $S$ is a solution set of edges from {\sc Diameter-2 Augmentation} on input $G_2$. We will show how to construct a proper augmenting set from $S$ of at most the same size as $S$. For any vertex $w \in U_1 \cup U_2$, there must be a 2-path (or less) joining $x$ to $w$. If such a 2-path ever passing through the vertex $z$, we can remove the $\{z,w\}$ edge from $S$ and add $\{x,w\}$ to $S$ instead. Note that such an edge-swap can never increase the diameter of the graph. We will provide a sequence of edge-swapping rules to the set $S$ until we arrive at a proper augmenting set.

\begin{edgerule}
If $S$ has an edge $\{z,w\}$ for any $w \in G_2$ then remove $\{z,w\}$ and add $\{x,w\}$.
\end{edgerule}

To describe the rest of the rules, we partition $U_1 \cup U_2$ into the following sets:

\begin{itemize}
\item[i)] $U_x$ = vertices $u$ in $U_1 \cup U_2$ such that $\{x,u\} \in S$
\item[ii)] $U^-$ = vertices $u$ in $U_1 \cup U_2$ that are not in $U_x$ and there is an edge $\{x,y(u,w)\} \in S$
\item[iii)] $U^+$ = vertices in $U_1 \cup U_2$ that are not in $U_x \cup U^-$
\end{itemize}

Clearly, these three sets are disjoint from each other and their union is exactly $U_1 \cup U_2$. To arrive at a proper augmenting set, the edges of $S$ joining vertex $x$ to the set $Y$ will have to be removed. It should be easy to verify that each of the following rules will not increase the diameter of $H$.

\begin{edgerule}
  If $S$ has an edge $\{x,y(a,b)\}$ with $a$ adjacent to $b$ then remove $\{x,y(a,b)\}$ and add the edge $\{x,a\}$.
\end{edgerule}

\begin{edgerule}
  If $S$ has an edge $\{x,y(a,b)\}$ with $a$ in $U_x$ then remove $\{x,y(a,b)\}$ and add the edge $\{x,b\}$.
\end{edgerule}

\begin{edgerule}
  If $S$ has edge $\{x,y(a,b)\}$ and $a$ is adjacent to some $c$ in $U_x$ then remove $\{x,y(a,b)\}$ and add the edge $\{x,b\}$.
\end{edgerule}
 
\begin{edgerule}
  If $S$ has two edges $\{x,y(a,b)\}$ and $\{x,y(b,c)\}$ then remove both of them and add the edges $\{x,y(a,c)\}$ and $\{x,b\}$.
\end{edgerule}

\begin{edgerule}
  If $S$ has two edges $\{x,y(a,b)\}$ and $\{x,y(c,d)\}$ such that $a$ is adjacent to $c$ in $G_2$ then remove $\{x,y(a,b)\}$ and $\{x,y(c,d)\}$ and add $\{x,a\}$ and $\{x,a(b,d)\}$.
\end{edgerule}

After applying Rules 3-6 we may have to return to Rule 2 and repeat this process, if any such edges would exist. Each rule reduces the number of edges from $x$ to the $Y$ set, so this process must indeed terminate.

Once we arrive at a point where none of the above rules can be applied any further, we make the following observations:

\begin{proposition}
  The set $U^-$ is empty.
\end{proposition}

\begin{proof}
  If any edge exists in $U^-$ then Rule 6 could be applied, so we have that $U^-$ is a stable set. If any edge existed from $U^-$ to $U_x$ then this would imply Rule 4 could be applied. Now consider any vertex $u$ in $U^-$: it must have an adjacent twin vertex, call it $u^t$, and it must be in $U^+$. Every vertex in $U^+$ must have a 2-path to $x$, but $U^+$ are the vertices which are not adjacent to any vertex in $Y$, and so every $U^+$ must be adjacent to one neighbour of $x$ in $U_x$. Now if $u^t$ is adjacent to some $a \in U_x$ then so is $u$, which violates Rule 4. Hence no such $u$ can exist, so $U^-$ is empty once these rules can no longer be applied.
\end{proof}

Proposition 1 tells us that all edges in the augmenting set $S$ must be from $x$ to $U_x$. We introduce one last rule to make this augmenting set proper:

\begin{edgerule}
  If $S$ has any $\{x,u\}$ edge where $u \in U_2$ then let $u^t$ be the twin of $u$ and remove $\{x,u\}$ and add the edge $\{x,u^t\}$.
\end{edgerule}

Now with a proper augmenting set, we can extract a dominating set of size at most $k$ in $U_1$. In the above notation, this is exactly the set $U_x$ when there are no more edge-swap rules that can be applied.

\end{proof}

\subsection{The Diameter-Improvement Problem}

Consider the following problem, which asks if the diameter of a graph can be improved (i.e. lowered):

\begin{problem}{\sc Diameter Improvement}
\\{\sc Input:} A graph $G=(V,E)$ and a positive integer $k$.
\\{\sc Task:} To determine if there exists a set of at most $k$ edges that can be added to $G$ so that the resulting graph has a smaller diameter than G.
\end{problem}

As previously noted, the graph resulting from the reduction from {\sc Dominating Set} to {\sc Diameter-2 Augmentation} had diameter 3 from its construction. Finding an augmenting edge set that improves this graph to diameter 2 will in fact solve the dominating set problem on the original (pre-reduction) graph. This provides a proof that {\sc Diameter Improvement} is itself W[2]-hard (and NP-complete,) even when restricted to input graphs of diameter 3.

\section{Concluding Remarks}

We gave a reduction to {\sc Diameter-2 Augmentation} directly from {\sc Dominating Set} which establishes the fixed-parameter hardness of {\sc Diameter-2 Augmentation} with respect to the augmenting set size. This also provides a proof of NP-completeness for {\sc Diameter-2 Augmentation} which reduced directly from a known and standard NP-complete problem. We identified the {\sc Diameter Improvement} and noted that it is fixed-parameter hard. Future considerations include finding exact exponential-time algorithms that are faster than brute-force searching for {\sc Diameter-2 Augmentation}, as well as the classification of subclasses of graphs for which {\sc Diameter-2 Augmentation} or {\sc Diameter Improvement} can be solved in polynomial time.

\end{document}